\pdfoutput=1
\documentclass[journal]{IEEEtran}
\usepackage[shortlabels]{enumitem}
\usepackage{titling}
\usepackage{stmaryrd}
\usepackage{palatino} 
\usepackage{mathpazo}
\usepackage{braket}
\usepackage{amsfonts}
\usepackage{amssymb}
\usepackage{amsmath}
\usepackage{mathrsfs} 
\usepackage{bbm}
\usepackage{latexsym}
\usepackage{amsthm}
\usepackage{complexity}
\usepackage{algorithm}
\usepackage[noend]{algpseudocode}
\usepackage[usenames]{color}
\usepackage{hyperref, etoolbox}
\usepackage{diagbox}
\usepackage{subcaption}
\usepackage{verbatim}
\usepackage{ragged2e, graphicx, pifont}
\usepackage{indentfirst}
\usepackage{thm-restate}	
\usepackage{caption}
\hypersetup{
colorlinks = true,
citecolor= blue
}

\usepackage{tikz-cd}
\usetikzlibrary{decorations.pathmorphing}
\usetikzlibrary{arrows.meta}

\usepackage[framemethod=TikZ]{mdframed}

\newcommand{\eps}{\varepsilon}

\newcommand{\clE}{\mathcal{E}}

\newcommand{\clS}{\mathcal{S}}

\newcommand{\Id}{\mathbb{1}}

\newcommand{\bbF}{\mathbb{F}}

\newcommand{\bbC}{\mathbb{C}}

\newcommand{\qbra}[1]{\llbracket{#1}\rrbracket}
\newcommand{\cbra}[1]{[{#1}]}

\renewcommand{\C}{\mathrm{C}}

\renewcommand{\H}{\mathrm{H}}

\newcommand{\Q}{\mathrm{Q}}

\newcommand{\supp}{\mathrm{supp}}

\newcommand{\wt}{\mathrm{wt}}
\newcommand{\dist}{\mathrm{dist}}
\newcommand{\qdist}{\mathrm{qdist}}
\newcommand{\Ddist}{\mathrm{Gdist}}
\newcommand{\mindeg}{\mathrm{mindeg}}
\newcommand{\Cl}{\mathrm{Cl}}
\renewcommand{\poly}{\mathrm{poly}}

\newcommand{\MD}{\textsc{MinDist}}
\newcommand{\GMD}{\textsc{GapDist}}
\newcommand{\GAMD}{\textsc{GapAddDist}}
\newcommand{\QMD}{\textsc{QMinDist}}
\newcommand{\GQMD}{\textsc{GapQDist}}
\newcommand{\GAQMD}{\textsc{GapAddQDist}}

\newcommand{\matel}[3]{\langle #1 | #2 | #3\rangle}

\newtheorem{theorem}{Theorem}
\newtheorem{lemma}{Lemma}

\newtheorem{cor}[theorem]{Corollary}

\newtheorem{fact}{Fact}
\newtheorem{claim}{Claim}

\theoremstyle{definition}
\newtheorem{definition}{Definition}
\newtheorem{problem}{Problem}

\title{On the hardness of the minimum distance problem of quantum codes}
\author{Upendra~Kapshikar\thanks{ U. Kapshikar is with Centre for Quantum Technologies, National University of Singapore.}~
        and~Srijita~Kundu
\thanks{S. Kundu is with Institute for Quantum Computing, University of Waterloo, Canada.}
}
\date{}

\begin{document}
\maketitle
\begin{abstract}
We study the hardness of the problem of finding the distance of quantum error-correcting codes. The analogous problem for classical codes is known to be NP-hard, even in approximate form. For quantum codes, various problems related to decoding are known to be NP-hard, but the hardness of the distance problem has not been studied before. In this work, we show that finding the minimum distance of stabilizer quantum codes exactly or approximately is NP-hard. This result is obtained by reducing the classical minimum distance problem to the quantum problem, using the CWS framework for quantum codes, which constructs a quantum code using a classical code and a graph. A main technical tool used for our result is a lower bound on the so-called graph state distance of 4-cycle free graphs.
In particular, we show that for a 4-cycle free graph $G$, its graph state distance is either $\delta$ or $\delta+1$, where $\delta$ is the minimum vertex degree of $G$.
Due to a well-known reduction from stabilizer codes to CSS codes, our results also imply that finding the minimum distance of CSS codes is also NP-hard.
\end{abstract}

\section{Introduction}
In 1995, Shor~\cite{Shor_coding} showed that similar to classical computation, quantum computation could also \emph{handle} errors using a quantum analogue of error-correcting codes.
The famous 9-qubit code he constructed could correct a single Pauli error on any of the nine qubits.
Soon after that, Calderbank and Shor~\cite{CS98} and Steane~\cite{steane} came up with a standard procedure to combine two classical error-correcting codes (satisfying certain conditions) to obtain a quantum error-correcting code. In his seminal work~\cite{Dan_thesis}, Gottesman formalized the \emph{stabilizer} setup, giving a group theoretic framework to the study of quantum error correction. 

As far back as the works of Calderbank, Shor and Steane, one of the most important directions for constructing new quantum codes has been the fruitful use of classical error correction. One common way to obtain quantum codes is by the CSS construction; using two classical codes such that one is contained in the dual of the other. Codeword stabilized (CWS) codes~\cite{CWS} present another way to use classical codes. Codeword stabilized codes are quantum codes made out of two classical objects: a graph and a classical error-correcting code. The class of CWS codes contains all the stabilizer codes and also encompasses some non-stabilizer (or non-additive) quantum codes.
\paragraph{Hardness of problems related to error correction.} The two most important computational problems related to error-correcting codes are: the problem of decoding a codeword upon which an error has acted and the problem of finding the minimum distance of a code. The decoding problem can be stated in many forms, but in syndrome decoding of linear codes, we are given an error \emph{syndrome}. In the classical case, if the parity check matrix of the code is $H$, the syndrome of an error $e$ on codeword $x$ is $H(x+e)=He$. In \emph{maximum likelihood decoding}, we are given an error syndrome $s$, and we want to find the most likely error that leads to it; if we assume there is an error in each component independently, this is the vector $u$ that satisfies $Hu=s$ and has the minimum Hamming weight. The minimum distance problem of a code looks similar to the maximum likelihood decoding problem: here, we are given a parity check matrix of the code and want to find a non-zero vector $u$ of minimum weight such that $Hu$ is the zero vector. Note that because of the non-zero vector requirement in the minimum distance problem, it is not simply a special case of the maximum likelihood decoding problem.

Both the maximum likelihood decoding problem and the minimum distance problem for classical codes (in their decision versions) are known to be NP-complete \cite{BMT78, Vardy}. The first problem being NP-complete, obviously closes the avenue for a generic decoding algorithm that works for all linear codes and errors of any weight. Of course, this problem is trying to decode all linear codes, under all possible errors: good decoding algorithms exist for specific classes of codes, and when the error is promised to have weight at most $d/2$, where $d$ is the distance of the code. However, finding the distance of a generic code in order in order to verify that the promise on the error is satisfied, is also hard. The minimum distance problem being NP-hard also rules out an obvious way of finding codes with good distance: if there was a polynomial time algorithm for the problem, one could generate parity check matrices at random and check if the associated code has good enough distance or not. Such an algorithm might still have been possible if one could efficiently compute the minimum distance of a code approximately rather than exactly. But in fact, (multiplicative and additive) approximate versions of the minimum distance problem have also been shown to be NP-hard under polynomial-time randomized reductions \cite{MDS07}.

Syndrome decoding for quantum stabilizer codes can be defined analogous to the classical case. In this case however, rather than finding the most likely error, finding the most likely error coset is desired \cite{PC08}. In \cite{HL11}, it was shown that using Hamming weight as the distance metric, decoding of stabilizer codes is NP-hard, regardless of whether the most likely error or most likely error coset is considered. In \cite{Fuj11}, it was shown that using symplectic weight as the distance metric, the decision version of the maximum likelihood decoding problem is NP-complete. Ref.~\cite{KL13} showed that this result holds even when one considers a smaller class of stabilizer codes with full-rank check matrices, and when one restricts the error model to the depolarizing channel. Ref.~\cite{IP13} considered the quantum decoding problem, while taking degeneracy into account: they showed that the problem of degenerate quantum maximum likelihood decoding, which tries to find the most likely equivalence class of errors, all of which lead to the same syndrome, is in fact $\#$P-complete.

As far as we are aware, the hardness of the quantum minimum distance problem has however not been studied. Note that showing that a minimum distance problem is hard can be quite different from showing a decoding problem is hard: in the classical case, there was a gap of nearly two decades between showing that the maximum likelihood decoding problem is hard and the minimum distance problem is hard. The results showing the NP-hardness of quantum decoding problems have all reduced the classical maximum likelihood decoding problem to the quantum problem. For example, in the formulation of \cite{Fuj11}, the decision version of the maximum likelihood decoding problem for stabilizer codes is stated as follows: given two $m\times n$ binary matrices $A$ and $A'$ such that $AA'^T+A'A^T$ is the zero matrix, a syndrome $y$ of length $m$, do there exist vectors $x$ and $x'$ of length $n$ such that the Hamming weight of $x \lor x'$ is at most $w$, and $Ax + A'x' = y$? The decision version of the classical problem (is there a vector $u$ of weight at most $w$ such that $Hu=s$?) reduces to this easily by setting $A=H, y=s$, and $A'$ being the zero matrix. The minimum distance problem for stabilizer codes can similarly be stated as follows: given matrices $A$ and $A'$ such that $AA'^T + A'A^T$ is the zero matrix, is there a vector $(x \vert x')$ outside the row-space of $( A \vert A^\prime)$, such that at least one of $x$ and $x'$ is non-zero, the Hamming weight of $x \lor x'$ (or equivalently, the symplectic weight of $(x \vert x')$) is at most $w$, and $Ax+ A'x'$ is the zero vector? However, it is not difficult to see that the classical minimum distance problem cannot be easily reduced to this problem. In particular, setting $A'$ to be the zero matrix does not help: if we set $A^\prime=0$, then the quantum problem always has a solution with symplectic weight 1 by setting $x^{\prime}$ to be any vector of Hamming weight 1, and $x=0$. 

Moreover, even though both CWS codes and CSS codes use classical linear error-correcting codes as their central building blocks, the hardness of the minimum distance of classical linear codes does not immediately imply the hardness of the minimum distance problem for CSS codes, or CWS codes.
For a CSS code $\Q_{CSS}$, comprised of two classical codes $\C_1$ and $\C_2$ (satisfying $\C_1^\perp \in \C_2$), one has the following relation: $\qdist(\Q_{CSS}) \geq \min\lbrace{\dist(\C_1),\dist(\C_2)}\rbrace$.
But note that, due to the orthogonality condition on CSS codes, one can not use an arbitrary pair $\C_1,\C_2$. 
If we want to reduce the classical minimum distance problem, starting with $\C_1$, we need to find a code $\C_2$, such that, $\C_2$ satisfies the orthogonality condition and has minimum distance not less than $\C_1$. 
One way to get around this is to use self-dual (or weakly self-dual) classical codes in the CSS construction. But it is not clear whether the hardness result for classical codes still holds under the restriction that the code is self-dual (or weakly self-dual).

In the above cases, the main bottleneck seems to be the following: one cannot use an arbitrary linear code in the construction of quantum codes that we want to show hardness for; they should come from some restricted subclass $\mathcal{C}_0$ of the class of all linear codes $\mathcal{C}$.
The CWS framework is free from any such restriction, in that, one can use any classical code $\C \subset \mathbb{F}_2^n$ and a graph $G$ on $n$ vertices to construct a quantum code.  
But this framework suffers badly while translating distance from $\C$ to $\Q$; even if one starts with linear code $\C$ having distance $\dist(\C)$, the constructed CWS code $\Q$ can be of distance $\qdist(\Q) \ll \dist(\C)$. 
\subsection{Our results and proof overview}
In this paper, we investigate the hardness of finding the minimum distance of a quantum code.
Our first result can be informally stated as below:
\begin{theorem} \label{thm:intro_CWS_hardness}
The minimum distance problem for stabilizer codes is NP-complete. Multiplicatively or additively approximating the distance of a stabilizer code is also NP-hard under polynomial-time randomized reductions. The problems remain NP-hard even with the promise that the code is non-degenerate.
\end{theorem}
In the proof of Theorem \ref{thm:intro_CWS_hardness}, we will be using CWS representation of stabilizer codes.
We consider a subclass of CWS codes where the classical code is linear. It is known that such a code is necessarily equivalent to a stabilizer code. 

Theorem \ref{thm:intro_CWS_hardness} is proved by reducing the classical minimum distance problem to the quantum minimum distance problem.
That is, given an arbitrary classical linear code $\C$ we want to construct a quantum code $\Q$, such that the distance of $\C$ is equal to that of $\Q$.
In the CWS framework, a quantum code $\Q=(G,\C)$ is constructed using a classical code $\C$ and a graph $G$. The distance of $\Q$ depends not only on $\C$ but also on $G$; hence an arbitrary choice of $G$ would not work for the reduction to the classical case. For example, by Fact~\ref{fc:qdist-ub}, if $\C$ uses all its components, then $\qdist(G,\C) \leq \min \lbrace \delta+1 , \dist(\C) \rbrace$ where $\delta$ is the minimum degree of $G$. The upper bound of $\delta+1$ comes from a property of the graph called its graph state distance $\Ddist(G)$ (that we define formally in Section \ref{sec:prelim}) that is relevant for error correction. For our reduction, it is useful if $\Ddist(G)$ is close to its maximum possible value of $\delta+1$.

In the following lemma, we give a sufficient condition on a graph $G$, so that its graph state distance is very close to its maximum value.
\begin{lemma} \label{lemma:intro_Ddist_restruct}
If $G$ is a graph with minimum degree $\delta$ and no 4-cycles, then $\Ddist(G)$ is either $\delta$ or $\delta+1$.
\end{lemma}
Using this lemma, and an appropriate family of graphs, we can construct a family of stabilizer codes such that $\qdist(\Q=(G,\C))=\dist(\C)$,
and thus prove Theorem~\ref{thm:intro_CWS_hardness}.

Since the distance of the quantum code we constructed in the reduction is equal to the distance of the classical code, the same reduction also works for the approximate versions of the respective problems. However, the additive factors in the additive approximation are different in the classical and quantum case: the additive approximation problem considered by \cite{MDS07} is to distinguish between the cases that a classical code has distance at most $t$ or at least $t + \tau\cdot n$, for $\tau > 0$. Whereas we consider the problem of distinguishing between the CWS code having distance at most $t$ or at least $t+\tau \cdot \sqrt{n}$. This is because the quantum problem we reduce the classical problem to will have distance $O(\sqrt{n})$, due to the graph having to be free of 4-cycles. Moreover, the randomized reductions for the approximate problems in  Theorem~\ref{thm:intro_CWS_hardness} are of the same type as those considered in the classical result \cite{MDS07}; we expand on this more in Section \ref{sec:prelim}. 

\subsubsection*{Organization of the paper}
The rest of the paper is organized as follows: In Section \ref{sec:prelim} we describe classical and quantum codes, in the CWS and stabilizer framework, in more detail, as well as state the exact and approximate versions of the minimum distance problem.
Section~\ref{sec:diag} is devoted to proving a lower bound on the graph distance, which, as stated before, will be our main tool for proving our results.
In section~\ref{sec:hardness}, we state and prove the formal versions of Theorem~\ref{thm:intro_CWS_hardness}. 

\addtocounter{theorem}{-1}
\addtocounter{lemma}{0}

\section{Preliminaries}\label{sec:prelim}
\subsection{Classical codes}
A linear error-correcting code $\C$ is a $k$-dimensional linear subspace of an $n$-dimensional space, for $k \leq n$. Elements of the code $\mathrm{C}$ are called \emph{codewords}.
In this paper we shall only consider vector spaces over $\bbF_2$, which is the alphabet of the code.
A codeword $c \in \C$ is, thus, a bit string of length $n$ and will be represented as $c_1,c_2,\ldots, c_n$.
We say that a code $\C$ uses all its components if for every $i \in [n]$, there exists a codeword $c \in \C$ such that $c_i \neq 0$.

The \emph{minimum distance}, or simply \emph{distance}, of a code is the minimum Hamming distance between two distinct codewords $u$ and $v$, i.e., the number of components in the vector $u\oplus v$ that have non-zero entries. For a linear code $\C$, the minimum distance is the same as the minimum Hamming weight of a non-zero codeword (note that the zero vector is always a codeword), and it will be denoted by $\dist(\C)$. A $k$-dimensional linear code with distance $d$, sitting inside an $n$-dimensional subspace is denoted as $\cbra{n,k,d}$ code (when we do not explicitly want to refer to distance, we shall call it an $[n,k]$ code).
Note that an error is a vector $e \in \mathbb{F}_2^n$ takes a codeword $x$ to $x\oplus e$. The code does not detect the error $e$ if $x\oplus e$ is also a codeword; since the code is linear, this means that an undetectable error $e$ is a (non-zero) codeword. Therefore, $\dist(\C)$ is in fact the minimum Hamming weight of an error that is not detected by the code. A code may not actually be able to correct all the errors that it detects, but it can correct errors of weight up to $\lfloor(\dist(\C)-1)/2\rfloor$.

Since a code $\C$ is a $k$-dimensional subspace of an $n$-dimensional space, it can be thought of as the kernel of a matrix $H \in \mathbb{F}_2^{(n-k)\times n}$, which is called its \emph{parity check matrix}. The code is completely specified by this matrix; we shall use $\C(H)$ to denote the code specified by the parity check matrix $H$. The distance of a code is then the minimum Hamming weight of a non-zero vector $u$ such that $Hu = 0^{n-k}$. We shall formally define the minimum distance problem for a classical code in these terms.

\subsection{Quantum codes}\label{sec:Qcodes-prelim}
An $(\!( n,K )\!)$ quantum error-correcting code $\Q$ is a $K$-dimensional subspace of $(\bbC^2)^{\otimes n}$. 
For the purpose of this work, we will consider $K=2^k$, unless stated otherwise. 
Similarly, the errors we consider are  Pauli operators.
The Pauli operators, given by
\[ \{X(a)Z(b) = X^{a_1}Z^{b_1}\otimes\ldots\otimes X^{a_n}Z^{b_n} : a, b \in \{0,1\}^n\},\]
where $X$ and $Z$ are the single qubit Pauli $X$ and $Z$ matrices, form a basis for matrices acting on $(\bbC^2)^{\otimes n}$. In particular, any error can be expanded in the Pauli basis. The support of a Pauli operator is the set of components $i$ such that $a_i \lor b_i=1$. For a general error, its support is the union of the supports of the Pauli operators in its expansion. The \emph{weight} of an error is then defined as the cardinality of its support.

Analogous to classical codes, the distance of a quantum code is defined as the minimum $d$ such that the code can detect errors of weight up to $d-1$, and correct errors of weight up to $\lfloor(d-1)/2\rfloor$. 
The Knill-Laflamme condition of error correction \cite{Laflamme} states that a Pauli error $E$ can be detected by a quantum code $\Q$ spanned by an orthonormal basis $\{\ket{c_i}\}_{j=1}^K$ iff for all $i, j$
\begin{equation}\label{eq:kl-cond}
\matel{c_i}{E}{c_j} = \delta_{ij}f(E)
\end{equation}
where $f$ is a function that depends only on the error $E$ (and not $\ket{c_i}, \ket{c_j}$). Therefore, the distance $\qdist(\Q)$ of a code $\Q$ is the minimum weight of an error $E$ that violates \eqref{eq:kl-cond}.\footnote{Note that there may still be some errors of weight greater than $\qdist(\Q)-1$ that are detectable by the code, just not all such errors.}
For $i\neq j$, equation \eqref{eq:kl-cond} ensures that an error acting on a codeword does not cause it to have overlap with another codeword, so that they can be perfectly distinguished --- this condition is analogous to what is required for classical codes. The condition for $i=j$, requiring that all diagonal elements of an error with respect to the codewords should be the same in order for the code to be detectable, is unique to quantum codes. The value $f(E)$ of the diagonal elements determines an important property of quantum codes known as \emph{degeneracy}.
\begin{itemize}
\item If $f(E)=0$ for all errors $E \neq \Id$ of weight up to $\qdist(\Q)-1$, then the code is called \emph{non-degenerate}.
\item If there exists some error $E \neq \Id$ of weight up to $\qdist(\Q) - 1$ for which $f(E) \neq 0$, then the code is called \emph{degenerate}.
\end{itemize}
In the rest of the paper, we shall only deal with CWS quantum codes. In these codes, the only allowed values of $f(E)$ are $0$ and $\pm 1$, and $f(E) \neq 0$. See Claim~\ref{claim:degen_in_CWS_via_Cl_G} in section~\ref{subsec:CWS} for more details.

\subsubsection{Stabilizer formalism and graph states \label{subsec:graph_states}}
Stabilizer codes are one of the most commonly studied classes of quantum codes.
A stabilizer code is a linear subspace of $(\bbC^2)^{\otimes n}$ that is the simultaneous $+1$ eigenspace of all elements of a \emph{stabilizer group}. A stabilizer group is an abelian subgroup of the Pauli group that does not contain $-\Id$. It can be shown that such a subgroup must always have a simultaneous $+1$ eigenstate.

A stabilizer code is completely specified by a minimal set of generating elements of its corresponding stabilizer group. If the size of a minimal generating set for the stabilizer group is $n-k$, then the corresponding stabilizer code is $2^k$-dimensional. Elements of the Pauli group are of the form $X(a)Z(b) = X^{a_1}Z^{b_1}\otimes\ldots X^{a_n}Z^{b_n}$, along with an overall phase of $\pm 1$ or $\pm i$.
Adding a $\pm 1$ or $\pm i$ phase to any stabilizer, changes the code to a code that is isomorphic to the original code.
Additionally, one can represent any element of the Pauli group (up to this isomorphism) as $X\left(a \right) Z\left(b\right)$~\cite[Section 7.9.2]{Preskill_lecture}.
Hence, we need not keep track of phases.
In other words, without loss of generality, one can just focus on stabilizer generators of the form $X(a)Z(b)$.
We shall represent each generator only by the string $(a|b)$. An $\qbra{n,k}$ stabilizer code is specified by a matrix $S \in \bbF_2^{(n-k)\times 2n}$, where each row of the matrix is interpreted as a string of the form $(a|b)$, and corresponds to a generator. Two rows $(a_1|b_1)$ and $(a_2|b_2)$ of this matrix are required to satisfy $a_2b_1-a_1b_2 = 0$, since the stabilizer group has to be abelian.
This representation is called the $\mathbb{F}_2$-representation or the symplectic notation \cite[Section 7.9.2]{Preskill_lecture}.

Another common way of obtaining a stabilizer code is via a graph. Let $G$ be a simple graph on $n$ vertices and with the adjacency matrix $A_G$ having columns $\lbrace u^{i}\rbrace_{i=1}^{n}$.
One can define the following set of Pauli matrices corresponding to the graph $G$;
\begin{align} \label{eq:standard_graph_stabilizer}
S_i= X_i Z^{u^i} & \hspace{5mm}  \text{ for } \hspace{2mm} 1 \leq i \leq n
\end{align}
Note that all the $S_i$ defined above commute, generating a stabilizer group.
Such a stabilizer group stabilizes a unique state $\ket{s}$, that is, there is a unique $\ket{s}$, such that $S_i\ket{s}=\ket{s}$ for all $i$. 
This $\ket{s}$ is called the graph state (corresponding to graph $G$). 
The distance $d^\prime$ of a graph state is the minimum weight of a non-trivial element of the stabilizer group $S= \langle S_1,S_2,\ldots, S_n\rangle$ \cite{KDP11}.
Thus a graph state is a one-dimensional non-degenerate stabilizer code with distance $d^\prime \leq \delta+1$, where $\delta$ is the minimum vertex degree of $G$. The upper bound on $d^\prime$ follows, as $S_{i}$ has weight $\deg(v_i)+1$. 
Given a graph $G$, stabilizer generators $S_i$ as given by equation~\eqref{eq:standard_graph_stabilizer}, involving a single $X$ operator and $Z$ according to  adjacency matrix of graph $G$, will be referred to as the \emph{standard form} of stabilizers for graph $G$.
Given a graph state $\ket{s}$, the set $\lbrace Z(a)\ket{s} \rbrace_{a \in \lbrace 0,1 \rbrace^n}$ forms an orthonormal basis known as the \emph{graph basis}.

\subsubsection{CWS codes\label{subsec:CWS}}

CWS codes form a larger class of quantum codes that includes stabilizer codes\cite{CWS}.
An $\qbra{n,k}$ CWS code is described by a graph $G$ with $n$ vertices and an $[n,k]$ classical error-correcting code $\C$. In general, the classical code part of a CWS code does not have to be a linear code. But in this work, we shall only concern ourselves with the case when $\C$ is linear. CWS codes with $\C$ being linear are exactly equivalent to stabilizer codes \cite{CWS}.

Formally, one can define CWS codes via any maximal stabilizer group $S$ and a set of $W$ of Pauli operators called as \emph{word operators}. Let $S$ be a maximal stabilizer group (of size $2^{n}$) and $W=\lbrace w_l \rbrace_{l=1}^{K}$ be $n$-qubit Pauli elements.
Let $\ket{s}$ be the unique state stabilized by $S$.
Then the code spanned by $\ket{w_l}$ defined as follows:
\[\ket{w_l}=w_l \ket{s}.\] 
As one wants codewords to be distinct,
only one $w_i$ can be in $S$.
Typically, it is chosen to be $w_1=\mathbb{1}$.
Using the local-Clifford equivalence of $S$ with graph state stabilizers~\cite{CWS}, one can define CWS code in the \emph{standard form} with a graph $G$ and furthermore, word operators can be chosen from a classical code $\C$. 
This gives the following definition, which we will be mostly concerned with throughout this work.
\begin{definition}[CWS$(G,\C)$] Let $G$ be a graph on $n$ vertices and $\C \subseteq \mathbb{F}_2^n$ be a classical code.
\begin{itemize}
\item Let $S_i$ be the stabilizer generators (see equation~\eqref{eq:standard_graph_stabilizer}) corresponding to the graph $G$, and $S$ be the the stabilizer group generated by them. 
Let $\ket{s}$ be the state stabilized by $S$.
\item Corresponding to each codeword $c \in \C$, define a word operator $w_c=Z(c)$.
\end{itemize} 
Then the CWS code CWS$(G,\C)$ is the code spanned by the vectors $\lbrace Z(c)\ket{s}\rbrace_{c \in \C}$.
\end{definition}

The idea behind CWS codes is that any Pauli error $X(a)Z(b)$ acting on a basis codeword of a CWS code is equivalent up to sign to an error of the form $Z(b')$ acting on the same codeword. These $Z$ errors can then be corrected by the classical error-correcting code part of the CWS code.\footnote{The `obvious' quantum encoding of a classical code as quantum states would be as computational basis elements. Then classical errors are actually of the form $X(a')$. However, if the encoding is done in the Hadamard basis instead, then the classical code can deal with errors of the form $Z(b')$.}

More specifically, let $\{u^i\}_{i=1}^n$ be the columns of the adjacency matrix $A_G$ of the graph $G$ associated with a CWS code. For a basis codeword $\ket{c}$ of a CWS code and $E=X(a)Z(b)$, we have
\begin{equation}\label{eq:CWS-E-sign}
E\ket{c}= \pm Z(\Cl_G(E))\ket{c}.
\end{equation}
The sign in the above equation depends on the error, but crucially, not on the codeword), and the function $\Cl_G$ is defined as
\begin{equation}
\Cl_G(E) = b \oplus \left(\bigoplus_{i=1}^na_iu^i\right).
\end{equation}
Intuitively speaking, this means that an $X$ error on the $i$-th qubit in $E$ propagates down the edges of the $i$-th vertex of the graph $G$, and acts as a $Z$ error on all the neighbours of the $i$-th vertex. Note that $\Cl_G(E)$ can be the zero vector, even if $E$ is a nontrivial Pauli error.

We shall use the following result about CWS codes proved in \cite{CWS}.
\begin{fact}[Theorem 3 in \cite{CWS}]\label{fc:CWS-detect}
A CWS code $(G,\C)$, where $\C$ is a classical code with codewords $\{c_j\}_j$ and $G$ is a graph whose adjacency matrix has columns $\{u^i\}_i$, detects a Pauli error $E$ if and only if the following conditions are met:
\begin{enumerate}[(i)]
\item If $\Cl_G(E) \neq 0^n$, then it must be detectable by $\C$ (i.e., it must not be a non-zero codeword of $\C$)
\item If $\Cl_G(E) = 0^n$, then $E$ must satisfy $Z(c_j)E = EZ(c_j) \, \forall j$.
\end{enumerate}
\end{fact}
Obviously, if the code detects Pauli errors from a set $\clE$, then it also detects errors that are linear combinations of Pauli errors from $\clE$. The distance of the code is then the smallest weight of Pauli error $E\neq \Id$ for which the conditions in Fact \ref{fc:CWS-detect} are not satisfied. We shall use this characterization of distance in most of our analysis in this paper.

For CWS codes, another important parameter is the  \emph{graph state distance} (or simply \emph{graph distance}) $\Ddist(\Q)$.
\begin{definition}
    Let $\Q= (\G,\C)$ be a CWS code. Then 
    \[\Ddist(\Q) = \min_{E \neq \Id} \lbrace \wt(E) \ : \ Cl_G(E)=0 \rbrace.\]
\end{definition}

Unlike $\qdist(\Q)$, the graph state distance of $\Q$ is a property determined entirely by its associated graph $G$. Hence sometimes we shall talk about $\Ddist(\Q)$ (without there necessarily being an associated code); with some abuse of notation, we shall also use $\Ddist(G)$ to refer to this. 
As the name suggests, the graph state distance  $\Ddist(G)$ matches the graph state distance defined in Section~\ref{subsec:graph_states} due to the following claim.
\begin{claim}
Let $G$ be a graph, with the corresponding stabilizer generators $S_i$ in the standard form. 
Let $d^\prime$ be the distance of graph state given by $G$.
Then
$d^\prime= \min_{E \neq \mathbb{1}} \lbrace \wt(E): \Cl_G{\left(E\right)}=0\rbrace$.
\end{claim}
\begin{proof}
We will show that $E =X(a)Z(b)$ is in $S$ if and only if $\Cl_G(E)=0$.
The claim then directly follows.
First, let us show this for $E=S_j= X(e^j) Z(u^j)$. \[\Cl_{G}(E)= u^{j} \oplus \bigoplus_{i=1}^{n} e^{j}_{i} u^{i} = u^{j} \oplus u^{j}=0.\]
The second equality follows since $e^{j}_{i}=\delta_{ij}$.
Thus, the claim holds for any $E=S_i$.
By linearity, one can extend this to any $E \in S$.
For this, note that if $E_1=X(a_1) Z_(b_1)$ and $E_2=X(a_2) Z(b_2)$ are such that $\Cl_{G}(E_1)=\Cl_G(E_2)=0$, then for $E_{3}=X(a_1\oplus a_2) Z(b_1 \oplus b_2)$, we get $\Cl_{G}(E_3)=0$. 
As any $E \in S$ can be written as $X(\oplus_{i\in {T_E}} e^i) Z(\oplus_{i \in {T_E}} u^{i})$ (for some $T_E \subset \lbrace 1,2,\ldots,n\rbrace$), this completes the proof.   
\end{proof}
The following claim relates quantity $f(E)$ in the Knill-Laflamme condition to $\Cl_{G}(E)$.
\begin{claim}
\label{claim:degen_in_CWS_via_Cl_G}
Let $\Q$ be CWS codes constructed from graph $G$ and code $\C$. 
Let $E$ be an error such that $\wt(E) \leq \qdist(\Q)-1$. Then,
\begin{enumerate}
\item $\Cl_G(E)=0$ if and only if $f(E) = \pm 1$.
\item $\Cl_{G}(E) \neq 0$ if and only if $f(E) = 0$ 
\end{enumerate}
In particular, only possible values of $f(E)$ are $\pm 1\text{ and }0$.
\end{claim}
\begin{proof}
Let $\ket{s}$ be the graph state corresponding to $G$, which we note is a codeword for the CWS code.
Since, $\wt(E)\leq \qdist(\Q)-1$, $E$ satisfies the Knill-Laflamme condition.  $\langle c\vert E \vert c\rangle$ thus has the same value for every codeword, and by definition, this value is equal to $f(E)$. In particular, we have,
\begin{align*}
    f(E)&= \langle s \vert E \vert s\rangle \\
    &= \pm \langle s \vert Z\left( \Cl_{G}(E)\right) \vert s\rangle,
\end{align*}
 where for the last equality we have used the fact that $E\ket{s}= \pm Z\left(\Cl_{G}(E) \right) \ket{s}$ by \eqref{eq:CWS-E-sign}. Suppose $\Cl_G(E)=0$, then $Z\left(\Cl_{G}(E)\right) = \mathbb{1}$ and hence, $f(E)=\pm 1$.
 Otherwise, $Z\left(\Cl_{G}(E)\right)= Z(a \neq 0)$. 
 Since $Z(a) \ket{s}$ forms an orthonormal graph basis, $\ket{s}$ is orthogonal to $Z(a)\ket{s}$ for any $a \neq 0$. Thus, in this case, we get $f(E)=0$.
\end{proof}
 
The following upper bound on graph state distance and distance of CWS codes holds for all graphs \cite{KDP11}; we provide a proof for completeness.

\begin{fact}\label{fc:diag-dist-ub}
Suppose $G$ is a graph with $\mindeg(G)=\delta$, and $\C$ is a classical code that uses all its components, i.e., $\forall i\in[n], \exists c \in \C$ s.t. $c_i \neq 0$. Then for the CWS code $\Q=(G,\C)$, $\qdist(\Q), \Ddist(\Q) \leq \delta + 1$.
\end{fact}
\begin{proof}
For an $i$ for which there exists $c\in \C$ s.t. $c_i \neq 0$, we shall show that there exists an error $E$ of weight $\deg(v_i)+1$ ($v_i$ being the $i$-th vertex in $G$) such that $\Cl_G(E)=0^n$, and $Z(c)E \neq EZ(c)$. By Fact \ref{fc:CWS-detect}, this proves the upper bound on both $\qdist(\Q)$ and $\Ddist(\Q)$.

Let $E=X(e^i)Z(u^i)$, where $u^i$ is the $i$-th row in $A_G$, corresponding to the neighbours of $v_i$; The weight of $E$ is $\deg(v_i)+1$, since $u^i$ has 1s in $\deg(v_i)$ locations, excluding the $i$-th location where $e^i$ has a 1. Clearly $\Cl_G(E)=u^i\oplus u^i=0^n$. Moreover, since $c$ has a 1 in the $i$-th location, $Z(c)$ anti-commutes with $X(e^i)$. Therefore we have, $Z(c)E = Z(c)X(e^i)Z(u^i) = -X(e^i)Z(c)Z(u^i) = -X(e^i)Z(u^i)Z(c) = -EZ(c)$. This completes the proof.
\end{proof}
Using condition (i) of Fact \ref{fc:CWS-detect}, the above fact then has the following corollary.
\begin{fact}\label{fc:qdist-ub}
Suppose $G$ is a graph with $\mindeg(G)=\delta$ and $\C$ is a classical code that uses all its components. Then for $\Q=(G,\C)$, $\qdist(\Q)\leq \min\{\delta+1,\dist(\C)\}$.
\end{fact}


\subsection{The minimum distance problem}

\begin{problem}[Classical minimum distance problem, $\MD$]\label{prob:cmindist}
\leavevmode

\begin{tabular}{p{1.5cm}p{5.6cm}}
\textsc{Instance:} & A matrix $H \in \bbF^{(n-k)\times n}$ and an integer $t > 0$ \\
\textsc{Yes:} & $\dist(\C(H)) \leq t$.\\
\textsc{No:} & $\dist(\C(H)) > t$\\
\end{tabular}
\end{problem}

It is clear that Problem \ref{prob:cmindist} is in NP; it was also shown to be NP-hard by Vardy.
\begin{fact}[\cite{Vardy}]\label{fc:c-NPcomp}
$\MD$ is NP-complete.
\end{fact}

In some applications, one may not need to exactly compute the distance of a code, but an approximation of it may be enough. Ref.~\cite{MDS07} studied two kinds of approximate versions of $\MD$, multiplicative and additive. In the multiplicative version, we wish to approximate the distance of a code up to a constant factor $\gamma$; in the additive version, we wish to approximate it by an additive factor of $\tau\cdot n$. The decision versions of these approximation problems are stated as the following promise problems.\footnote{In \cite{MDS07}, the problems are stated in terms of the generator matrix of the code, instead of the parity check matrix. The generator matrix is a matrix whose row space is equal to the code subspace, and given the generator matrix one can efficiently compute the parity check matrix and vice versa. Hence the two representations are equivalent for our purposes.}

\begin{problem}[Multiplicative approximate minimum distance problem, $\GMD_\gamma$]
\leavevmode

\begin{tabular}{p{1.5cm}p{5.6cm}}
\textsc{Instance:} & A matrix $H \in \bbF^{(n-k)\times n}$, an integer $t > 0$ and an approximation factor $\gamma \geq 1$.\\
\textsc{Yes:} & $\dist(\C(H)) \leq t$.\\
\textsc{No:} & $\dist(\C(H)) > \gamma \cdot t$.\\
\end{tabular}
\end{problem}

\begin{problem}[Additive approximate minimum distance problem, $\GAMD_\tau$]
\leavevmode

\begin{tabular}{p{2cm}p{5.6cm}}
\textsc{Instance:} & A matrix $H\in \bbF^{(n-k)\times n}$, an integer $t> 0$, and an approximation factor $\tau > 0$.\\
\textsc{Yes:} & $\dist(\C(H)) \leq t$.\\
\textsc{No:} & $\dist(\C(H)) > t + \tau\cdot n$.\\
\end{tabular}
\end{problem}

It was shown in \cite{MDS07} that $\GMD$ and $\GAMD$ are NP-hard under polynomial-time \emph{reverse unfaithful randomized} (RUR) reductions. These are probabilistic reductions where no-instances always map to no-instances, and yes-instances map to yes-instances with high probability. In particular, given a parameter $s$, the \cite{MDS07} reduction maps a yes-instance correctly with probability $1-2^{-s}$ in time $\poly(s)$.

\begin{fact}[\cite{MDS07}, Theorems 22 and 32]\label{fc:c-approx-hard}
For every $\gamma \geq 1$, $\GMD_\gamma$ is NP-hard under polynomial time RUR reductions with soundness error exponentially small in a security parameter. Moreover, there exists a $\tau > 0$ such that $\GAMD_\tau$ is NP-hard under polynomial time RUR reductions with the same soundness error.
\end{fact} 
Note that the above results imply that there is no RP algorithm $\GMD_\gamma$ or $\GAMD_\tau$, with any $\gamma$ and the $\tau$ given by the above fact, unless NP = RP. In fact the result for $\GMD$ can be amplified to get that for every $\eps > 0$, there is no RQP algorithm for $\GMD_\gamma$ with $\gamma = 2^{\log^{(1-\eps)}n}$, unless NP is contained in RQP (class of problems having randomized quasipolynomial time algorithms with one-sided error).

We define the quantum versions of the above classical problems analogously. This will be done in the framework of CWS codes described in \ref{sec:Qcodes-prelim}. We note that any graph $G$ with $n$ vertices and $[n,k]$ classical code describes a valid CWS code. Therefore, any parity-check matrix $H \in \bbF_2^{(n-k)\times n}$ and a graph $G$ with $n$ vertices (along with parameter $t$) forms a valid instance of the exact minimum distance problem of a CWS code. The two approximate versions are of course promise problems just like in the classical case.
\begin{problem}[Quantum minimum distance problem, $\QMD$]\label{prob:qmindist}
\leavevmode

\begin{tabular}{p{1.5cm}p{5.6cm}}
\textsc{Instance:} & A matrix $H \in \bbF_2^{(n-k)\times n}$, a graph $G$ with $n$ vertices, and an integer $t > 0$ \\
\textsc{Yes:} & $\qdist(\C(H),G) \leq t$ i.e., is there a Pauli error $E\neq \Id$ of weight $\leq t$ that does not satisfy the conditions of Fact \ref{fc:CWS-detect}.\\
\textsc{No:} & $\qdist(\C(H),G) > t$\\
\end{tabular}
\end{problem}

\begin{problem}[Multiplicative approximate quantum minimum distance problem, $\GQMD_\gamma$]\label{problem:gqmindist_multiplicative}
\leavevmode

\begin{tabular}{p{1.5cm}p{5.6cm}}
\textsc{Instance:} & A matrix $H \in \bbF_2^{(n-k)\times n}$, a graph $G$ with $n$ vertices, an integer $t > 0$, and an approximation factor $\gamma \geq 1$.\\
\textsc{Yes:} & $\qdist(\C(H),G) \leq t$.\\
\textsc{No:} & $\qdist(\C(H),G) > \gamma \cdot t$.
\end{tabular}
\end{problem}

\begin{problem}[Additive approximate quantum minimum distance problem, $\GAQMD_\tau$]\label{problem:gqmindist_additive}
\leavevmode

\begin{tabular}{p{1.5cm}p{5.6cm}}
\textsc{Instance:} & A matrix $H \in \bbF_2^{(n-k)\times n}$, a graph $G$ with $n$ vertices, an integer $t > 0$, and an approximation factor $\tau > 0$. \\
\textsc{Yes:} & $\qdist(\C(H), G) \leq t$.\\
\textsc{No:} & $\qdist(\C(H), G) > t + \tau\cdot\sqrt{n}$.
\end{tabular}
\end{problem}

Like $\MD$, it is easy to see that $\QMD$ is in NP. The two approximation problems, being promise problems, are not. The additive approximation problem $\GAQMD_\tau$ is defined with the additive factor being $\tau\cdot\sqrt{n}$ instead of $\tau\cdot n$ like in the classical case because that is what naturally arises in our reduction between the classical and quantum minimum distance problems. We can also define versions of these problems where the CWS code is promised to be non-degenerate, though we shall not explicitly name them.


\section{Lower bound on graph distance distance for 4-cycle free graphs}\label{sec:diag}
In this section, we prove Lemma~\ref{lemma:intro_Ddist_restruct}, a lower bound on the graph distance.
Before proving this lemma, we shall introduce some terms. Note that for a graph $G$ with {adjacency matrix $A_G$ having} columns $\{u^i\}_i$,
\[ b\oplus \left(\bigoplus_{i=1}^na_iu^i\right) = \left(\bigoplus_{j=1}^nb_je^j\right)\oplus\left(\bigoplus_{i=1}^na_iu^i\right) \]
where $e^j$ denotes the vector with 1 only in the $j$-th location. The above expression is the sum of a subset of columns of the matrix $(\Id|A_G)$. If $G$ has minimum degree $\delta \geq 2$, then each column of $(\Id|A_G)$ has Hamming weight either 1 or at least $\delta$ (with there existing a column which has weight exactly $\delta$). We call such a set of vectors a set with a $\delta$ \emph{degree gap}. If $\Cl_G(E)=0^n$ for some $E$, then that means a subset of columns of $(\Id|A_G)$ of size some $d$ (which corresponds to the weight of $E$ being between $d$ and $d/2$) sums to the zero vector. This subset is obviously one with a degree gap $\delta$.

It is not difficult to see that a graph $G$ has no 4-cycles if and only if each pair of vertices has at most one common neighbour. For a vector $v$, let $\supp(v)$ denote the set of components where $v$ has a non-zero entry. The condition that each pair of vertices has at most one neighbour can be stated in terms of the columns of $A_G$ as
\begin{equation}\label{eq:supp-cond}
|\supp(u^i) \cap \supp(u^j)| \leq 1
\end{equation}
for $j\neq i$. Equation \eqref{eq:supp-cond} is also true for the columns of $(\Id|A_G)$, since all the columns of $\Id$ have support size 1 anyway. We call a set of vectors where all pairs satisfy \eqref{eq:supp-cond} an at-most-one-matching (ATOM) set.

A subset of columns of $(\Id|A_G)$ that sums to zero is an ATOM set with degree gap $\delta$. We prove the following lemma for such a set of vectors, which further {implies} Lemma \ref{lemma:intro_Ddist_restruct}.

\begin{lemma}\label{lem:conds}
Let $S$ be a non-empty set of binary vectors which is ATOM, has degree gap $\delta$ and its elements sum to zero. Let $S_1$ denote the set of vectors of weight 1 in $S$, and let $S_\delta = S\setminus S_1$ be the set of vectors that have weight at least $\delta$. Then we must have, $\max\{|S_1|,|S_\delta|\} \geq \delta$.
\end{lemma}
\begin{proof}
The proof will be done by case analysis. 
{Our strategy here will be to exhaustively go over all possible sizes of $S$.
{For each case, we present a relevant claim. The statements and proofs of these claims can be found in the appendix.}
\begin{itemize}
    \item Claim~\ref{clm:1} shows that no  $S$ of size strictly less than $\delta+1$ satisfy the premise of Lemma. Furthermore, there are only two possible subsets $S$ of size $\delta+1$ that can satisfy the conditions of Lemma~\ref{lem:conds}, which the claim explicitly characterizes. Hence, Lemma~\ref{lem:conds} is proved for $\vert S \vert \leq \delta+1$.
    \item Claim~\ref{clm:2} deals with $\vert S \vert \geq 2\delta$, for which the lemma is trivially true.
    \item After the above two claims, the only cases left are $\delta+2 \leq \vert S \vert \leq 2\delta-1$. 
    The proof for this is further divided into 2 parts.
    \begin{itemize}
        \item Claim~\ref{clm:3} deals with the cases when $S_\delta$ contains a vector $v$ whose support is disjoint with the support of $S_1$.
        \item Claim~\ref{clm:4} deals with the case where such a $v$ does not exist, that is, every $v \in S_\delta$ has non-trivial intersection with the support of $S_1$.
    \end{itemize}
\end{itemize}}
Claims \ref{clm:1} to \ref{clm:4} cover all cases of Lemma \ref{lem:conds}. This completes the proof.
\end{proof}

\begin{proof}[Proof of Lemma \ref{lemma:intro_Ddist_restruct}]
Let $S$ satisfying the conditions of Lemma \ref{lem:conds} correspond to an error $E=X(a)Z(b)$. $S_1$ containing a vector $e^i$ means that $b_i=1$; similarly, each vector in $S_\delta$ corresponds to $a_j=1$ for a unique location $j$. If $\Cl_G(E) = 0^n$, then by Lemma \ref{lem:conds}, there are at least $\delta$ locations where $a_i=1$, or $\delta$ locations where $b_i=1$. Since the weight of $E$ is the number of locations where $a_i\lor b_i=1$, this means that the weight of $E$ is at least $\delta$. Thus the minimum weight of $E$ for which $\Cl_G(E)=0^n$, i.e., $\Ddist(G)$, is at least $\delta$.
\end{proof}

\section{NP-hardness of quantum minimum distance\label{sec:hardness}}

In this section, we prove our central result on the hardness of the minimum distance of quantum codes.
Our proof strategy will be as follows: we shall reduce the minimum distance problem for classical codes to the minimum distance problem in CWS codes. First, we shall go from the given $[n,k]$ classical code $\C$ to an $[m, k]$ code $\C'$ such that $m = O(n^2)$ and $\dist(\C')=\dist(\C)$. This is just to ensure that $\dist(\C') = O(m^{1/2})$. Moreover, we shall also pick $m$ such that $m=p^2+p+1$ for some prime $p$. Such an $m$ can be efficiently found due to Lemma~\ref{lem:mp-algo}.
{Once we have $\C'$, we shall use it to construct a CWS code $\Q = (\C',G)$ such that $\dist(\C)= \dist(\C')= \dist(\Q)$. The graph $G$ we use in $\Q$ will have no 4-cycles, so that we can apply Lemma \ref{lemma:intro_Ddist_restruct}. But we shall also need the graph to have a high enough degree for each vertex. We shall use Lemma~\ref{lem:4cycle-free} to construct such a graph. The construction is the so-called orthogonal polarity graph, defined by Erd\"{o}s, R\'{e}nyi and S\'{o}s \cite{ERS66}. This completes our reduction from the minimum distance of classical codes to  the minimum distance of CWS codes.}

\begin{restatable}{lemma}{mpalgo}
\label{lem:mp-algo}
For $n > 7$, there exists a number of the form $p^2+p+1$ in the interval $[n, 7n]$ for some prime $p$. Furthermore, given $n$, there exists an algorithm that finds this number in time $\poly(n)$. 
\end{restatable}
\begin{proof}
Let $q$ be the largest prime such that $q^2+q+1 < n$; such a prime always exists for $n> 7$. Let $p$ be the next prime after $q$. By Bertrand's postulate, $p$ must lie in the interval $[q+1, 2(q+1)]$. Now consider the number $p^2+p+1$. It must satisfy $p^2+p+1\leq 4(q+1)^2 + 2(q+1) + 1 < 7(q^2 + q + 1) \leq 7n$. Moreover, by definition $p^2+p+1 \geq n$. Therefore, $p^2+p+1$ is a number of the required form in the internal $[n,7n]$.

In order to find $p$, the algorithm first finds the largest prime $q$ such that $q^2 + q + 1 < n$. This can be done by checking all the numbers up to $\sqrt{n}$, which takes $\poly(n)$ time. Then the algorithm finds the next prime $p$ and outputs it. This can be done by checking all the numbers up to $2(q+1)$, which can also be done in $\poly(n)$ time.
\end{proof}

\begin{restatable}{lemma}{cyclefree}
\label{lem:4cycle-free}
For $m=p^2+p+1$ for some prime $p$, there is an algorithm that, given $m$, constructs a graph $G$ with $m$ vertices satisfying the following properties:
\begin{enumerate}[(i)]
\item Each vertex of $G$ has degree $p$ or $p+1$;
\item $G$ is free of 4-cycles.
\end{enumerate}
Furthermore, the running time of the algorithm is $\poly(m)$.
\end{restatable}
\begin{proof}
We use the construction due to Erd\"{o}s, R\'{e}nyi and S\'{o}s~\cite{ERS66}. Consider the set of elements $(x,y,z)$ in $\bbF_p^3 {\setminus \lbrace 0^3 \rbrace}$. We define an equivalence relation on these elements as follows: $(x,y,z)$ and $(\lambda x, \lambda y, \lambda z)$ for all $\lambda \neq 0$ are in the same equivalence class. The number of such equivalence classes is $p^2+p+1$, since we can represent each equivalence class by elements of the form $(1, y, z)$, $(0, 1, z')$ or $(0,0,1)$ for arbitrary $y, z, z' \in \bbF_p^3$. The number of elements of the first type is $p^2$, the number of elements of the second type is $p$, and there is only one element of the third type. The vertices of the graph $G$ will correspond to these $p^2+p+1$ equivalence classes.

Two vertices $(x, y, z)$ and $(x', y', z')$ are connected by an edge in $G$ iff they satisfy the condition
\begin{equation}\label{eq:colinear}
xx' + yy' + zz' = 0.
\end{equation}
(Note that if the pair of representative points of two equivalence classes satisfy this condition, then so does any other pair of points from the two classes.) Some vertices will satisfy $x^2+y^2+z^2=0$; we do not add self-loops for these vertices. We claim that the degree of vertices which do not satisfy $x^2+y^2+z^2=0$ is $p+1$, which would make the degree of the rest of the vertices $p$. To see this, we calculate the number of possible $(x',y',z')$ in $\bbF_p^3\setminus\{0^3\}$ which satisfy \eqref{eq:colinear} for a fixed $(x,y,z)$ first. Without loss of generality, assume $x\neq 0$; the argument is similar for the case when $x=0$ but $y \neq 0$ or $z \neq 0$. All points of the form $\left(-x^{-1}\cdot(yy'+zz'), y', z'\right)$ where either $y'\neq 0$ or $z'\neq 0$ are solutions to \eqref{eq:colinear}. This gives $p^2-1$ possible solutions in $\bbF_p^3\setminus\{0^3\}$. However, since $(\lambda x', \lambda y', \lambda z')$ for $\lambda\neq 0$ satisfies \eqref{eq:colinear} whenever $(x',y',z')$ does, we have overcounted by a factor of $p-1$. Therefore, the number of equivalence classes of solutions (including possibly the equivalence class of $(x,y,z)$ itself) is $p+1$, which makes the degree of vertices not satisfying $x^2+y^2+z^2=0$, $p+1$.

Finally, to prove that $G$ has no 4-cycles, we use the fact that a necessary and sufficient condition for a graph to be 4-cycle free is for no two vertices to have more than one common neighbour. For a vertex $(u, v, w)$ to be a common neighbour of two different vertices $(x, y, z)$ and $(x', y', z')$, it must satisfy
\begin{equation}\label{eq:2-neighbour}
\begin{pmatrix} x & y & z \\ x' & y' & z'\end{pmatrix}\begin{pmatrix} u \\ v \\ w\end{pmatrix} = \begin{pmatrix} 0 \\ 0 \end{pmatrix}.
\end{equation}
Note that the matrix on the right-hand side of \eqref{eq:2-neighbour} is rank 2, since $(x,y,z)$ and $(x',y',z')$ belong to different equivalence classes. Therefore, it has a null space of dimension 1, and up to the equivalence class, \eqref{eq:2-neighbour} can only have a single solution. Therefore, two vertices have exactly one common neighbour.

We have shown that the graph $G$ satisfies the required properties, and now we only have to show that the algorithm constructing it runs in time $\poly(m)$. This is easy to see: given $m=p^2+p+1$, the algorithm lists the representative points of all the equivalence classes, and for each pair of points, checks whether \eqref{eq:colinear} is satisfied or not, and adds an edge accordingly. This can be done in $\binom{m}{2}$ time.
\end{proof}

For $\Q = (\G,\C')$ constructed using Lemmas \ref{lem:mp-algo} and \ref{lem:4cycle-free}, we shall show that $\qdist(\Q) = \dist(\C')=\dist(\C)$, and hence if we can solve the minimum distance problem for $\Q$, we can also solve the minimum distance problem for $\C$, either exactly or approximately.

\begin{theorem}\label{thm:q-NPcomp}
$\QMD$ is NP-complete. For every $\gamma \geq 1$, $\GQMD_\gamma$ is NP-hard under polynomial time RUR reductions with soundness error exponentially small in a security parameter. There exists $\tau > 0$ such that $\GAQMD_\tau$ is NP-hard under polynomial time RUR reductions with the same soundness error. Versions of $\QMD, \GQMD_\gamma$ and $\GAQMD_\tau$ where the code is promised to be non-degenerate are also NP-hard.
\end{theorem}

\begin{proof}
Given an instance of $\MD$ or $\GMD_\gamma$ or $\GAMD_\tau$ consisting of $H\in \bbF_2^{(n-k)\times n}$ (with $n > 7$), first by applying Lemma \ref{lem:mp-algo} we can find $m \in [25n^2, 175n^2]$ which is of the form $p^2+p+1$. Consider the parity check matrix {$H' = H \oplus \mathbb{I}_{m-n}$}. Each codeword in $\C(H')$ is of length $m$, and is a codeword of $\C(H)$ appended with $m-n$ 0s. It is clear that $\dist(\C(H)) = \dist(\C(H')) \leq \frac{m^{1/2}}{5}$. Moreover, all codewords of $\C(H')$ have Hamming weight at most $\frac{m^{1/2}}{5}$.

Next, we construct a graph $G$ with $m$ vertices using Lemma \ref{lem:4cycle-free}. Consider the CWS code $\Q = (\C(H'), G)$. We shall first prove that $\qdist(\Q) \leq \dist(\C(H))$. Let $c$ be a non-zero codeword of $\C(H')$ with $\wt_\H(c) = \dist(\C(H))$. Now consider the error $E=Z(c)$, which has weight $\dist(\C(H))$. Clearly $\Cl_G(E) = c$, which $\C(H')$ cannot detect.

To prove $\qdist(\Q) \geq \dist(\C(H))$, we need to show that all Pauli errors of weight up to $\dist(\C(H)) - 1$ are detected by $\Q$, i.e., they satisfy the conditions of Fact \ref{fc:CWS-detect}. Note that it is sufficient to show that for such an $E$, $\Cl_G(E)$ is not a codeword of $\C(H')$. Firstly, for such an error, $\Cl_G(E)$ cannot be the zero codeword. This is because by Lemma \ref{lemma:intro_Ddist_restruct}, we have that the minimum weight of $E$, for which $\Cl_G(E)=0^m$, is $p \geq \frac{m^{1/2}}{\sqrt{2}}$, whereas $\dist(\C(H)) \leq \frac{m^{1/2}}{5}$.

To show that $\Cl_G(E)$ is also not a non-zero codeword, we shall divide the errors into two cases: first when $E$ is of the form $Z(a)$, and second when it is not. For an error of weight $\dist(\C(H))-1$ of the form $Z(a)$, $\Cl_G(E)=a$, whose weight is too low to be a codeword of $\C(H')$ by definition.
In the second case, $E$ has $X$ in some set of components $S$, such that $|S| \leq \dist(\C(H))-1$.
Suppose $S$ contains a component $i$. 
Then $\Cl_G(E)$ is of the form
\[ u^i \oplus a \oplus \left(\bigoplus_{j\in S, j \neq i}u^j\right)\]
for some $a$ such that $\wt_\H(a) \leq \dist(\C(H))-1$. The $u^j$-s are columns of $G$, and so they have 1s in $\geq \frac{m^{1/2}}{\sqrt{2}}$ components each. Moreover, we have for $j\neq i$,
\[ |\supp(u^i)\cap \supp(u^j)| \leq 1, \]
which means that
\begin{equation}\label{eq:common-1s}
\sum_{j \in S, j \neq i}|\supp(u^i)\cap \supp(u^j)| \leq \dist(\C(H)) - 1.
\end{equation}
$\Cl_G(E)$ may have 0 at a component at which $u^i$ has a 1 if $a$ has a 1 in that component, or one of the other $u^j$-s has a 1 in that component. Hence by \eqref{eq:common-1s}, we have that
\begin{align*}
\wt_\H(\Cl_G(E)) & \geq \wt_\H(u^i) - \wt_\H(a) - \dist(\C(H)) + 1 \\
 & \geq \frac{m^{1/2}}{\sqrt{2}} - 2(\dist(\C(H)) - 1) \\
 & \geq \frac{4m^{1/2}}{15}.
\end{align*}
Therefore, the weight of $\Cl_G(E)$ is too high to be a codeword of $\C(H')$.

Thus we have shown that $\dist(\C(H))=\qdist(\Q)$. Given an instance $(H, t)$ of $\MD$, our polynomial time construction gives us an instance of $\QMD$ with $H', G$ as described, and the same parameter $t$. If $(H, t)$ is an instance of $\GMD_\gamma$, $(\Q=(H',G),t)$ satisfies the promise $\qdist(\Q)\leq t$ or $> \gamma \cdot t$. and is therefore an instance of $\GQMD_\gamma$. If $(H, t)$ is an instance of $\GAMD_\tau$, then $(\Q=(H',G),t)$ satisfies the promise $\qdist(\Q) \leq t$ or $> t + \tau\cdot\sqrt{m}$. By Facts \ref{fc:c-NPcomp} and \ref{fc:c-approx-hard}, this completes the proof of the first part. The hardness of the problems when restricted to non-degenerate codes follows because the quantum code $\Q$ we have constructed satisfies $\qdist(\Q) \leq \Ddist(\Q)$, and hence is non-degenerate.
\end{proof}

\subsection{CSS codes and other stabilizer code representations}
Hardness of problems related to quantum error correction has so far been studied in the \emph{standard} stabilizer formalism, where the input is given as ( $\mathbb{F}_2$ or $\mathbb{F}_4$ representation of) generators of the stabilizer group.
A stabilizer code is described by the generators of its corresponding stabilizer group.
As mentioned before, the problems that we show hardness results for are all stated in the CWS framework.
However, the CWS description of a code can be efficiently converted to its stabilizer description~\cite{KDP11}.
Therefore, our hardness results from Theorem~\ref{thm:intro_CWS_hardness} translate to the standard input of the stabilizer framework as well.

Moreover, \cite{BTL10} gives an efficient way to map an $\llbracket n,k \rrbracket$ stabilizer code with distance $d$ to a $\qbra{4n,2k}$ CSS code with distance $2d$.

\begin{fact}[\cite{BTL10}, Corollary 1]
There is a mapping that takes an $\qbra{n,k}$ stabilizer code with distance $d$ and produces an $\qbra{4n,2k}$ CSS code with distance $2d$, in time $\poly(n)$.
\end{fact}
This mapping goes via an intermediate object called a Majorana fermion code, which we shall not describe in this paper. But the mapping allows us to reduce the (exact or approximate) minimum distance problem in stabilizer codes to the minimum distance problem in CSS codes. Hence our hardness result (Theorem~\ref{thm:intro_CWS_hardness}) has the following corollary.
\begin{cor}
The minimum distance problem for CSS codes is NP-hard. Multiplicatively or additively approximating the distance of a CSS code is also NP-hard under polynomial-time randomized reductions.
\end{cor} 
Just like in the classical case, the results for the approximate minimum distance problems imply that there is no RP (class of problems having polynomial-time randomized algorithms with one-sided error) algorithm for approximating the minimum distance of a stabilizer or CSS code.

\section{Conclusion}
In this work we have shown that the minimum distance problem for quantum codes is NP-hard, and moreover, multiplicatively or additively approximating the distance of a quantum code is also NP-hard under reverse unfaithful randomized reductions. Our hardness result is shown for stabilizer codes, and moreover, we have shown that the hardness remains if we restrict the problem to CSS codes, or provide the promise that the stabilizer code is promised to be non-degenerate. Our main tool is proving this result was a lower bound on the graph distance of 4-cycle free graphs. We obtain the result by reducing the classical problems to the quantum problems in the CWS framework, using a 4-cycle free graph.

One drawback of our result for additively approximating the distance of a quantum code is the fact that we consider an additive approximation factor of $\tau\cdot\sqrt{n}$, whereas the additive factor in the equivalent classical result is $\tau\cdot n$ ($n$ being the dimension of the code). This difference is because we use 4-cycle free graphs in our reduction, which cannot have minimum degree $\omega(\sqrt{n})$, and correspondingly the quantum code we construct from a given classical code cannot have $\omega(\sqrt{n})$ distance. One way to improve this result would be to find a different set of sufficient conditions that would let us prove a result similar to Lemma \ref{lemma:intro_Ddist_restruct}. That is, is there a graph $G$ having minimum degree $\Theta(n)$, such that $\Ddist(G)$ is still lower bounded by the minimum degree? If such a graph exists, it could be used in place of the 4-cycle free graph we used in the proof of Theorem \ref{thm:q-NPcomp}, and obtain a better additive term in the result for additive approximation results. As long as we also use a classical code with linear distance, such a graph could also be used to get a simple construction of a quantum code with linear distance in the CWS framework.

\appendix
\section*{Claims in Proof of Lemma~\ref{lem:conds}}
\begin{claim}\label{clm:1}
For any $S$ satisfying the conditions of Lemma \ref{lem:conds}, must have $|S| \geq \delta + 1$. $S$ of size $\delta+1$ is of one of the following forms:
\begin{enumerate}
\item $S_\delta=\{v\}, S_1 = \cup_{j\in\supp(v)}\{e^j\}$
\item $S_\delta = \{v^1, \ldots, v^{\delta+1}\}$ for vectors $v^i$ of weight exactly $\delta$, $S_1 =\emptyset$.
\end{enumerate}
\end{claim}
\begin{proof}
All the weight 1 vectors in $S_1$ have 1s in different locations, and these 1s cannot cancel between themselves in order for the sum of vectors in $S$ to be zero. Therefore, we must have $|S_\delta|\geq 1$. If $S_\delta$ has exactly one vector, this vector has 1s in at least $\delta$ locations, all of which have to be cancelled by a unique vector in $S_1$. This gives $S_1=\cup_{j\in\supp(v)}\{e^j\}$, i.e., $|S_1| \geq \delta$. This means $|S|\geq \delta+1$. $|S|=\delta+1$ is achieved when the vector in $S_\delta$ has weight exactly $\delta$.

For the second part, we shall prove that if $|S|=\delta+1$ and $|S_\delta| \geq 2$, then $S_1 = \emptyset$. When $|S_\delta|\geq 2$, consider two vectors $v^1$ and $v^2$ in $S_\delta$. The vector $v^1\oplus v^2$ has 1s in at least $2\delta-2$ locations, since each of $v^1$ and $v^2$ has 1s in at least $\delta$ locations, of which at most one location is common. For $S$ to sum to zero, all these 1s need to be cancelled by other vectors in $S$, with each other vector being able to cancel at most two locations (one due to $v^1$ and one due to $v^2$). If each vector cancels exactly 2 locations, then an additional $\delta-1$ vectors in $S$ can take the sum of $S$ to zero. However, if $S$ has a vector of weight 1, the weight 1 vector is able to cancel a 1 in only one location, and hence the sum cannot be zero. Therefore, in this case $S$ cannot have any weight 1 vectors, i.e., $S_1=\emptyset$. Moreover, if any vector in $S_\delta$ has weight more than $\delta$, then the 1s cannot be cancelled by the other vectors in $S_\delta$. Therefore, all the vectors in $S_\delta$ must have weight exactly $\delta$.
\end{proof}
\begin{claim}\label{clm:2}
If $S$ satisfies the conditions of Lemma \ref{lem:conds} and $|S| \geq 2\delta$, then $\max\{|S_1|,|S_\delta|\}\geq \delta$.
\end{claim}
\begin{proof}
Since, $S$ is the disjoint union of $S_1$ and $S_{\delta}$, at least one of them must be bigger than $\delta$, if $S$ is bigger than $2\delta$. 
\end{proof}
\begin{claim}\label{clm:3}
If $S$ satisfies the conditions of Lemma \ref{lem:conds} and contains a vector $v$ of weight more than 1 such that $\supp(v) \cap \supp(e^i) = \emptyset$ for all vectors $e^i \in S_1$, then $|S_\delta| \geq \delta+1$.
\end{claim}
\begin{proof}
The vector $v$ has 1 in at least $\delta$ locations, and no weight 1 vector has a 1 in any of these locations. For the vectors in $S$ to sum to zero, the 1s in $v$ have to be cancelled by 1s from other vectors. Since $S$ is ATOM, any vector can cancel only a single 1 in $v$. Since all the vectors cancelling the 1s in $v$ come from $S_\delta$, this means that $|S_\delta|\geq \delta+1$.
\end{proof}

We are now left with the task of proving Lemma \ref{lem:conds} for sets $S$ that do not fall under any of Claims \ref{clm:1}, \ref{clm:2}, \ref{clm:3}. This means that we need to consider $\delta+2\leq |S| \leq 2\delta-1$. In the proof of Claim \ref{clm:1} we already saw that $|S_1|\geq \delta$ if $|S_\delta|=1$ (if the single vector in $S_\delta$ has weight exactly $\delta$, this corresponds to $S$ being of size exactly $\delta+1$; if the weight of this vector is higher, then it corresponds to $S$ being of size $\delta+2$ or more). Therefore, we only need to consider $|S_\delta|\geq 2$. For $|S|\geq \delta+2$, if $|S_1|\leq 1$, then we have $|S_\delta|\geq \delta+1$, and therefore, we are already done. Therefore, we shall also consider $|S_1|\geq 2$. We shall also assume that $\delta\geq 3$, since if $\delta=2$ and $|S|\geq \delta+2$, we are already in the $|S|\geq 2\delta$ case of Claim \ref{clm:2}, and therefore done. Finally, we can assume that the condition in Claim \ref{clm:3} that $S$ contains a vector $v$ of weight more than 1 such that $\supp(v) \cap \supp(e^i) = \emptyset$ for all vectors $e^i \in S_1$, then $|S_\delta| \geq \delta+1$, is not true.

\begin{claim}\label{clm:4}
If $S$ is a set satisfying the conditions of Lemma \ref{lem:conds}, and additionally, the following are true:
\begin{enumerate}
\item $\delta \geq 3$, and $\delta+2 \leq |S| \leq 2\delta-1$;
\item $ |S_1|, |S_\delta| \geq 2$;
\item For each vector $v^j \in S_\delta$, there is a vector $e_j\in S_1$ such that $|\supp(v_j) \cap \supp(e_j)| = 1$.
\end{enumerate}
Then $\max\{|S_1|,|S_\delta|\} \geq \delta+1$.
\end{claim}
\begin{proof}
The proof will be by induction on the difference between the degree gap of the set $S$ and its size. The base case is $|S|=\delta+2$, where the difference is 2. The proofs for the base case and the induction step are quite similar, so we shall describe them together. For the induction step where the difference is $i+1$, the induction hypothesis is that $\max\{|S'_1|,|S'_{\delta'}|\} \geq \delta'+1$ for all sets $S'$ of size up to $\delta'+i$ and degree gap $\delta'$, satisfying the conditions of the lemma (the hypothesis is obviously not true for $i=1$, since in that case $\max\{|S'_1|,|S'_{\delta'}|\}$ may be $\delta'$ or $\delta'+1$). 

Suppose every vector in $S_\delta$ intersects (i.e., has a 1 in a common location with) the same vector in $S_1$. Since $S_1$ contains at least two vectors, this means there must be a vector in $S_1$ that does not intersect with any vector in $S_\delta$. But a 1 in a vector in $S_1$ cannot be cancelled by a 1 in another vector in $S_1$, therefore this cannot happen. Therefore, there must exist distinct vectors $v^i, v^j \in S_\delta$ and $e^i, e^j \in S_1$ such that $v^i$ intersects with $e^i$ and $v^j$ with $e^j$. We can also assume that $v^i$ is the highest weight vector in $S$. We shall divide the proof into two cases: when $v^i$ has weight $\delta$ (this means that all vectors in $S_\delta$ have weight $\delta$), and when it has weight more than $\delta$.

\paragraph{Case (i): $\wt_\H(v^i) = \delta$.} Consider the vectors $v^i\oplus e^i$ and $v^j\oplus e^j$. They have weight $\delta-1 \geq 2$, so they cannot be in the set $S$. We define the set $S' = \{v^i\oplus e^i, v^j\oplus e^j\}\cup S \setminus \{v^i, v^j, e^i, e^j\}$. $S'$ is ATOM, has degree gap $\delta'=\delta-1$, and sums to zero if and only if $S$ sums to zero. The size of $S'$ is 2 less than that of $S$, and therefore, the difference between its size and degree gap is 1 less than that of $S$.

If $|S|=\delta+2$, then $|S'|=\delta'+1$, and by Claim \ref{clm:1} it either has $|S'_{\delta'}|=1$ and $|S'_1|=\delta'$, or $|S'_{\delta'}|=\delta'+1$ and $|S'_1|=0$. But we in fact know the second case cannot happen as it would mean all the vectors in $S'_{\delta'}$ have weight $\delta'$, whereas we know two vectors in $S'_{\delta'}$ have weight $\delta'$, and the rest have weight $\delta'+1$. Therefore we have $|S'_1|=\delta'$, which means that $|S_1|=\delta'+2=\delta+1$.

If $|S|=\delta+i+1$, then $|S'|=\delta'+i$. Then by the induction hypothesis, $\max\{|S'_1|,|S'_{\delta'}|\} \geq \delta'+1$. Now $|S_1|=|S'_1|+2$ and $|S_\delta|=|S'_{\delta'}|$. This means that if $|S_1|\geq |S_\delta|-2$, then $|S_\delta|\geq \delta'+3 =\delta+2$.\footnote{Note that we do not assume in the statement of Lemma \ref{lem:conds} that $S$ is an ATOM set with degree gap $\delta$ of minimum size that sums to zero. Therefore, the lower bound of $\delta+2$ here does not contradict Fact \ref{fc:diag-dist-ub}. Rather it shows that sets which satisfy this case are never the sets of minimum size.} Otherwise, if $|S_\delta|+2 > |S_1|$, then $|S_\delta|\geq \delta$. But we shall show that the $|S_\delta|=\delta$ case cannot happen. If $|S_\delta|$ is of size $\delta$, then every vector in $S_\delta$ must have at least one 1 in some location that cannot be cancelled by 1s from other vectors in $S_\delta$. Moreover, this location must be unique for each vector in $S_\delta$, because otherwise two vectors from $S_\delta$ would intersect in more than one location. Therefore, there must be a unique vector in $S_1$ can cancels this 1 for each vector in $S_\delta$. This gives $|S_1| \geq \delta$, which means $|S| \geq 2\delta$. This contradicts the assumption that $|S| \leq 2\delta-1$, and therefore we must have $|S_\delta| \geq \delta+1$ in this case.

\paragraph{Case (ii): $\wt_\H(v^i) > \delta$.} Consider the vector $v^i\oplus e^i$. It has weight $\geq \delta$ and intersects with $v^i$ in at least $\delta \geq 3$ many places. Therefore it cannot be in $S$: otherwise $S$ would not be ATOM. Consider $S'=\{v^i\oplus e^i\}\cup S\setminus\{v^i, e^i\}$. $S'$ is ATOM, has degree gap $\delta$, and sums to zero if and only if $S$ sums to zero. The size of $S'$ is 1 less than that of $S$, and therefore, the difference between its size and degree gap is 1 less than that of $S$.

If $|S|=\delta+2$, then $|S'|=\delta+1$, and it satisfies either $|S'_1|\geq \delta$, or $|S'_{\delta}| \geq \delta+1$. In the first case we have $|S_1|=|S'_1|+1 \geq \delta+1$, and in the second case we have $|S_\delta|=|S'_\delta| \geq \delta+1$.

If $|S|=\delta+i+1$, then $|S|=\delta+i$. By the induction hypothesis, $\max\{|S'_1|, |S'_\delta|\}\geq \delta+1$. This gives either $|S_1|\geq \delta+2$, or $|S_\delta|\geq \delta+1$.
\end{proof}

\section*{Acknowledgement}
U.~K. is supported by the NRF grant NRF2021-QEP2-02-P05. S.~K. is funded by the NSERC Canada Discovery Grants Program and Fujitsu Labs America; research at the Institute for Quantum Computing (IQC) is supported by Innovation, Science and Economic Development (ISED) Canada.

\ifCLASSOPTIONcaptionsoff
  \newpage
\fi

\bibliographystyle{IEEEtran}
\bibliography{dia}

\end{document}